\documentclass[draft]{llncs}
\usepackage{amsmath}

\newcommand{\Bbb}{{\bf}}
\newcommand{\bK}{\mathbf K}
\newcommand{\bR}{\mathbf R}
\newcommand{\rank}{\mathrm{rank}}
\newcommand{\be}{\beta}
\newcommand{\cA}{\mathcal A}
\newcommand{\cG}{\mathcal G}
\newcommand{\de}{\delta}
\newcommand{\bA}{\mathbf A}
\newcommand{\bB}{\mathbf B}
\newcommand{\bD}{\mathbf D}
\newcommand{\bC}{\mathbf C}
\newcommand{\bG}{\mathbf G}
\newcommand{\bF}{\mathbf F}
\newcommand{\mtr}[4]%
{\left(\begin{array}{cc}#1 & #2\\#3 & #4\end{array}\right)}
\newcommand{\mtrc}[2]%
{\left(\begin{array}{c}#1 \\#2 \end{array}\right)}
\newcommand{\cM}{\mathcal M}
\newcommand{\di}{\mathrm{diag}}
\newcommand{\balpha}{\boldsymbol{\alpha}}
\newcommand{\LCM}{\mathrm{LCM}}
\newcommand{\DEN}{\mathrm{DEN}}

\newcommand{\bx}{\mathbf x}
\newcommand{\ra}{\rangle}
\newcommand{\la}{\langle}
\newcommand{\hbx}{\hat{\mathbf x}}
\newcommand{\cC}{\mathcal C}
\newcommand{\si}{\sigma}
\newcommand{\al}{\alpha}
\newcommand{\hx}{\hat {x}}
\newcommand{\cD}{\mathcal D}

\spnewtheorem*{defi}{Definition}{\bfseries}{\rmfamily}

\begin{document}

\title{Effective Matrix Methods\\ in Commutative Domains
  \thanks {This paper was  published in the book
   Formal Power Series and Algebraic Combinatorics, 
   (Ed. by D.Krob, A.A.Mikhalev, A.V.Mikhalev),
   Springer, 2000, 506-517.
  No part of this materials may be reproduced, stored in retrieval system,
  or transmitted, in any form without prior permission of the copyright owner.
  }
}
\author{Gennadi~I.~Malaschonok}
\institute{Tambov State University, 392622 Tambov, Russia\\
\email{malaschonok@math-iu.tambov.su}}
\maketitle

\begin{abstract}
Effective matrix methods for solving standard linear algebra
problems in  a commutative domains are discussed. Two of them
are new. There are a methods for computing adjoined matrices and
solving system of linear equations in a commutative domains.
\end{abstract}

\section{Introduction}
   Let $ \bR $ be a commutative domain with identity, $ \bK $ the
   field of quotients of $ \bR $. We assume that
  $\,\bR\,$ is equipped with an algorithm allowing
  {\em exact division\/}. This means that if two elements
  $\,a\,$ and $\,b\,$ of $\,\bR\,$ are given ($a$ being
  different from zero) such that $\,b=ac\,$ with
  $\,c\in \bR\,$, then this algorithm can exhibit the exact
  quotient $\,c\,$. Let $\,{\bR}^{n \times m}\,$ denote
  the set of $\,n\times m\,$ matrices with entries in
  $\,{\bR}\,$.

   This paper is devoted to the review of effective matrix methods in
   the domain $\bR$  for a solution of standard linear algebra problems.
   They are (1)~multiplicating two matrices, (2)~solving  
   linear systems in $\bK$, (3)~solving linear systems in $\bR$, 
   (4)~computing the adjoint  matrix, (5)~computing the matrix
   determinant, (6)~computing the characteristic polynomial of a
   matrix.

   We shall estimate algorithms according to the  total number of
   multiplication and division operations in the ring $\bR$.

\smallskip
   (1). {\bfseries Multiplication of two matrices}. Let $ O(n^{\beta}) $ 
   be the number of  multiplication operations, that are
   necessary for multiplication of square matrices of the order $n$.
   For the standard multiplication of matrices $ \beta =3 $, for
   Strassen \cite{Str1969} algorithm $ \beta = \log 7 $, and for the best
   on today algorithm $ \beta < 2.376 $ \cite{CoppWin1990}.

\smallskip
   (2). {\bfseries Solving linear systems  in $K$}.
   Let $ Ax = b $ be a systems of linear equations whose coefficients
   belong to the commutative domain  $ \bR $: $ A\in \bR^{n\times m} $, 
   $ b\in \bR^{n} $, $x\in\bK^m$.
   The main method here is the so-called Gauss method with exact
   divisions with complexity $O(n^2 m)$ operations in $\bR$. 
   First this method was published
   in the paper of Dodgson \cite{Dod1866}, and later it was developed 
   in the works  \cite{Bar1968},  \cite{SM1968}, \cite{Mal83},
   \cite{Mal91},   \cite{Mal95}.
  We adduce the asymptotic complexity of these methods: $1.5n^2m$ 
  for Dodgson
  method \cite{Dod1866}, $1.5n^2m$ for Bareiss method \cite{Bar1968},
  $n^3$ for forward and back-up procedures \cite{Mal83}, $2/3n^2m$ for 
  one-pass method \cite{Mal91}, $7/12n^2m$ for generalized method
  \cite{Mal95}. A fast method of solving systems of linear equations 
  over commutative domain is published in the article  \cite{Mal97}.
  The complexity of this method is $O(n^{\beta-1} m)$, the same as 
  the complexity of matrix  multiplication.

\smallskip
  (3). {\bfseries Solving linear systems in $R$}. Let $ Ax = b $ be a 
  system of linear equations $ A\in \bR^{n\times m} $, $ b\in \bR^{n} $.
  The problem is to find all the solutions  $x$ of this system in the 
  module $\bR^m$.
  The particular cases of this problem are discussed in \cite{Mal86},
  \cite{Mal87}.
  A randomized algorithm for finding all the solutions in $R$ is discassed
  in  the section 3 of this paper.
It is supposed that there exists an
algorithm that is able to ascertain whether the finitely generated ideal
$I=(a_1,a_2,\ldots,a_s)$ is unit or not. If $I$ is unit, then
there are calculated the coefficients $k_i\in \bR$ in the expansion
of the unit $1=\sum_{i=1}^s k_i a_i$. It is possible to take in the
capacity of such algorithm the algorithm of the  Gr\"obner bases computation
in $\bR$.

\smallskip
(4). {\bfseries Computing the adjoint  matrix}.
The best known method of computing the adjoint matrix in an arbitary
commutative ring has the complexity 
$O(n^3 {\sqrt n} \, {\log} n \, {\log} {\log} n)$
operations of addition, substraction and multiplication \cite{Kalt92}.
If in a commutative ring the exact division is possible, then the best method 
has the complexity $O(n^3)$ operations of 
multiplication and exact division \cite{Mal87}, \cite{Mal91}.
In this work in section 2 we suggest the method with the complexity, equal to the complexity of matrix multiplication, i.e. $O(n^{\beta})$. 

\smallskip
(5). {\bfseries Computing the matrix determinant}.
The intermediate result of each algorithm \cite{Dod1866}, 
\cite{Bar1968}, \cite{Mal83}, \cite{Mal91}, \cite{Mal95}, \cite{Mal97}
for solving systems in $\bK$ is the computation of the matrix determinant.
So the asymptotic complexity of determinant computation for these methods is
$1.5n^3$, $1.5n^3$, $n^3$, $2/3n^3$, $7/12n^3$ respectively. For the 
methods \cite{Mal97} the complexity of determinant computation is $O(n^\beta)$.
The best method of computing the determinant of a matrix without divisions
was published by Kaltofen \cite{Kalt92}. Complexity of this method is
$n^3{\sqrt{n}} \log n \log \log n$.

\smallskip
 (6). {\bfseries Computing the characteristic polynomial of a matrix}.
In the case of an {\em arbitrary commutative ring\/},
the best algorithms are the Chistov one
\cite{chi}, and the Improved Berkowitz Algorithm
\cite{jou2} with size $O(n^{\beta + 1} \log {n})\,$.
In the paper \cite{AbMal}  there are described two new efficient
methods with $\,O(n^3)\,$ ring operations
(addition, subtraction, multiplication and exact
division). The first one is the Quasi-triangular  method
(with asymptotic  multiplicative  complexity $5/3n^3$) and
the second one is the Tri-diagonal method (with asymptotic 
 multiplicative  complexity $3n^3$).
As in the case of Hessenberg's method \cite{Fad63}, they proceed by
reducing the given matrix $A$ to a particular upper
quasi-triangular (Hessenberg) form, similar to $A$.

\medskip
{\bfseries Commutative domain of principal ideals.}
This is the basic application field.
In section 4 we discuss the problem of solving linear systems in 
the principal ideals domain $\bR$ and in the field of fractions $\bK$
Let $ Ax = c $ be a system of linear equations, $ A\in \bR^{n\times m} $,
$ c\in \bR^{n} $. 

\medskip
{\em{Solving linear system in the field of fractions.}}
The best known on today method for solving determined system $Ax=c$ in 
the field of fractions in the case  when $\bR=\Bbb Z$, $m=n$, $\det A \neq 0$,
is the Dixon method \cite{Dixon1982}, which uses the linear p-adic lifting.
Its complexity is $O(n^3 (\log n + \log \|A\| +\log p)^2 )$ bit operations.

If $m>n$, $\bR=\Bbb Z$ and Gauss method with exact divisions is used then
solving  system $Ax=c$ in usual arithmetic needs
$O(m n^4 (\log n + \log \|A\|)^2 )$ bit operations. $\|A\|$ denotes
the absolute value of the greatest coefficient of the system.
Using the Chinese remaindering method may reduce the complexity
up to $O(mn^3 (\log n + \log \|A\|)^2 )$ bit operations \cite{Bar1972}.

My approach to this problem uses p-adic lifting
like  in \cite{Dixon1982}. The complexity of the algorithm in the
case of the ring $\Bbb Z$ is 
$O((m-n+1)n^2 m(\log m +\log \|A\| +\log p)^2)$ bit operations.

\medskip
{\it{Solving linear system in the principal ideal domain.}}
In  \cite{MulStor99} there is given the randomized algorithm
for finding one solution of a system in the domain $\bR$, in the cases
when $\bR=\Bbb Z$ and $\bR=F[x]$ (F[x] --- ring of polynomials over a field).
This method is based on the Dixon algorithm. 

I suggest a randomized algorithm for finding all the solutions of a 
system in a commutative domain of principal ideals. It is based on using
p-adic lifting. Its complexity is essentially cubic in the
dimension of system like \cite{MulStor99},
but the number of matrix inversions is now $m-n$ times less.

\section{Adjoint Matrix Computation}
\subsection {Introduction}
Let
${\mathcal A}= \left(\begin{array}{cc} A&C\\B&D \end{array}\right)$
be the invertible matrix and $A$ --- its invertible block.
It is possible to factorize its inverse matrix ${\mathcal A}^{-1}$:
$$\left(\begin{array}{cc} I&-A^{-1}C\\0&I \end{array}\right)
      \left(\begin{array}{cc} I&0\\0& (D-BA^{-1}C)^{-1} \end{array}\right)
      \left(\begin{array}{cc} I&0\\-B&I                 \end{array}\right)
 \left(\begin{array}{cc} A^{-1}&0\\0&I                  \end{array}\right).
 \eqno (2.1) $$
Let ${\mathcal A}$ be a matrix of the order $n=2^p$. If a
block inversion by formula (2.1) is possible for the blocks up to the
second order, then the computation of inverse matrix needs
$2^{p-1}$ second order block inversions and $6 \cdot 2^{p-k-1}$ multiplications
of blocks the order $2^k$ $(k=1,2,\ldots,p-1)$.

Using the Strassen \cite{Str1969} algorithm of matrix multiplication
for such matrix inversion we need $(21 n^\be -6 n)/5$
multiplicative operations where $\be=\log_2 7$.

The similar method is proposed in this section.

\subsection{The Theorems of Adjoint Matrix Factorization}
Let $R$ be a commutative ring, ${\mathcal A}=(a_{i,j})$---a~square 
matrix of an order $n$ over the ring $R$. We denote by
$A^{k}_{i,j}$
its square submatrix of the order $k$, obtained by the bordering
of upper left block of order $k-1$ by the row $i$ and the column $j$, $(i,j>k)$.
Its determinant is denoted by $a^{k}_{i,j}={\det} A^{k}_{i,j} $.
Denote the corner minor of the order $k$ by $\delta_k=a^{k}_{k,k}$.
The determinant of the matrix, obtained from the upper left block
$A^{k}_{k,k} $ of order $k$
by the replacement of the column $i$ by the column $j$ is denoted by
$\delta_{k(i,j)}$, $(1\leq i \leq k, k<j\leq n)$.

Consider the matrices
$$
{\mathcal A}^{(s)}_t=(a^{s}_{i,j})^{i=s,\ldots,t}_{j=s,\ldots,t}
\ \ \mbox{ and}   \ \
{\mathcal G}^{(t)}_s=(\delta_{t(i,j)})^{i=s,\ldots,t}_{j=t+1,\ldots,n}
\eqno(2.2)
$$
of the order $(t-s+1)\times(t-s+1)$ and $(t-s+1)\times(n-t)$, correspondingly.

With the preceding notation the determinant Sylvester identity \cite{AlkMal96}
may be written in the following way:
$$
{\det} {\mathcal A}^{(s)}_t = \delta_{s-1}^{t-s} \delta_t
\eqno(2.3)
$$
where $1<s<t\leq n$.

Let us prove the two main theorems of the adjoint matrix factorization.
\begin{theorem}
Let a square matrix $\cA$ of the order $n$ over the ring $R$ be decomposed
into the blocks
$$
{\cA}= \left(\begin{array}{cc} A&C\\B&D \end{array}\right),
$$
$A$ is the square block of order $s$, $(1<s<n)$,
whose determinant $\delta_s$ is not zero or zero divisor in 
$R$. Then the adjoint matrix ${\cA}^{*}$ can be factorized
$$
\left(\begin{array}{cc}
    \delta_s^{-1}\delta_n I &-\delta_s^{-1}FC \\ 0&I \end{array}\right)
\left(\begin{array}{cc} I&0\\0&G           \end{array}\right)
\left(\begin{array}{cc} I&0\\-B&\delta_s I \end{array}\right)
\left(\begin{array}{cc} F&0\\0&I           \end{array}\right),
\eqno (2.4) $$
where
$F=A^{*}$, $G=\delta_s^{-n+s+1}{\mathcal A}_n^{(s+1)*}$, $I$ is the unit matrix
and the following identity takes place:
$$
{\mathcal A}_n^{(s+1)}=\delta_s D -BFC. \eqno (2.5)
$$
\end{theorem}

\begin{theorem}
Let the square matrix $\cA_n^{(s+1)}$ of order $n-s$, $(s>0$,\linebreak
\typeout{linebreak}
$n-s>2)$,
over the ring $R$ be decomposed into the blocks
$$
\cA_n^{(s+1)}= \left(\begin{array}{cc} \bA&\bC \\ \bB&\bD \end{array}\right),
$$
where $\bA$ is the square block of the order $t-s$, $(s<t<n)$,
$\delta_s$ and $\delta_t$ are not zeros or zero divisors in
$R$. Then the matrix $\de_{s}^{-n+s+1}\cA_n^{(s+1)*}$
can be factorized:
$$
\left(\begin{array}{cc}
    \de_t^{-1}\de_n I & -\de_t^{-1}\bF\bC \\ 0&I \end{array}\right)
\left(\begin{array}{cc} I&0   \\ 0    & \de_s^{-1}\bG \end{array}\right)
\left(\begin{array}{cc} I&0   \\ -\bB & \de_t I       \end{array}\right)
\left(\begin{array}{cc} \bF&0 \\ 0    & I             \end{array}\right),
\eqno(2.7) $$
where
$\bF=\de_{s}^{-t+s+1}\cA_t^{(s+1)*}$,
$\bG=\delta_t^{-n+t+1}\cA_n^{(t+1)*}$, $I$ is the unit matrix and
the following identity is true:
$$
\cA_n^{(t+1)}=\de_s^{-1}(\de_t\bD -\bB\bF\bC) \eqno(2.8)
$$.
\end{theorem}

\begin{proof}
Calculate the products of matrix $\cA_n^{(s+1)}$ by the
factors of the matrix (2.7) step by step from the right to the left:
$$
\cA_n^{(s+1)}\rightarrow
\left(\begin{array}{cc} \de_t I &\bF\bC \\ \bB&\bD \end{array}\right)
\rightarrow
\left(\begin{array}{cc} \de_t I &\bF\bC \\ 0& \de_t \bD-\bB\bF\bC
                                                          \end{array}\right)
\rightarrow
\left(\begin{array}{cc} \de_t I &\bF\bC \\ 0& \de_n I     \end{array}\right)
\rightarrow
\de_n I .
$$
It is necessary to prove the identity (2.8) and the following identities:
$$
\bF\bA=\de_t I \eqno(2.9)
$$
$$
\de_t^{-n+t+1}\cA_n^{(t+1)*}\cA_n^{(t+1)}=\de_n I  \eqno(2.10)
$$
As $\bA= \cA_t^{(s+1)}$, the equality (2.9) follows from the determinant
Sylvester identity
$
{\det} {\mathcal A}^{(s+1)}_t = \delta_{s}^{t-s-1} \delta_t
$.

The identity (2.10) follows from the determinant Sylvester identity
$${\det} \cA^{(t+1)}_n = \delta_{t}^{n-t-1} \delta_n.$$
Let us prove the identity (2.8).
Denote by $a_{l,j}^{(s+1)*}$ the cofactor of the element $(l,j)$
in the matrix $\cA^{(s+1)}_t$. From the determinant Sylvester identity
we obtain:
$$
\de_s^{t-s-1}\de_{t(i,j)}=\sum_{l=s+1}^t a_{l,i}^{(s+1)*} a_{l,j}^{s+1}.
$$
Since  $ \bF= \de_s^{-t+s+1}\cA^{(s+1)*}_t$ and $\bC$ is the block of the
matrix $\cA^{(s+1)}_n$,  the last equality for the elements implies
the matrix identity $\cG_t^{(s)}=\bF\bC$.
We decompose the determinant of the matrix $A_{i,j}^{s+1}$
according to the last column and obtain
$$
a_{i,j}^{s+1}=a_{i,j}\de_s-\sum_{l=1}^s a_{l,j}\si_{s(l,j)},
\eqno (2.11)
$$
where $\si_{s(l,j)}$ is the determinant of the matrix,
resulted from $A_{s,s}^{s}$ by the replacement of the row $l$ by the row $j$.
Let
$
\si=(\si_{s(1,i)},\si_{s(2,i)},\ldots,\si_{s(s,i)},0,0,\ldots,0),
$
$
\alpha=(a_{1,j},a_{2,j},\ldots,a_{t,j}),
$
$
\beta=(a_{i,1}^{s+1},a_{i,2}^{s+1},\ldots,a_{i,t}^{s+1})
$
denote the rows with $t$ elements. Then according to (2.11) we obtain
the matrix identity
$$
\left(\begin{array}{cc}  I & 0 \\ -\si & \de_s \end{array}\right)
\cdot A_{i,j}^{t+1} =
\left(\begin{array}{cc}
A_{t,t}^{t} & \alpha^T \\
\beta         & a_{i,j}^{s+1}
\end{array}\right).
$$
Correspondingly, we write the following determinant identity,
where the determinant of the matrix on the right  is decomposed
according to the last row:
$$
\de_s a_{i,j}^{t+1}=\de_t a_{i,j}^{s+1} -
\sum_{l=1}^t a_{i,l}^{s+1} \de_{s(l,j)}.
$$
In the matrix form it may be written as
$
\de_s\cA_n^{(t+1)}=\de_t\bD -\bB\cG_t^{(s)}.
$
Taking into account $\cG_t^{(s)}=\bF\bC $
we get the identity (2.9).
\end{proof}

\subsection{The Estimate of the Complexity }
The dimension of the upper left block $A$ in the process of the factorization
of the matrix may be taken arbitrarily. Consider the case, when the dimension
of the block $A$ is a degree of 2.
We call such decomposition of the adjoin matrix
the~{\em binary factorization}.

Let ${\it M}(n)$ be the complexity of the multiplication of
two matrices of the order $n$ and its asymptotical estimate is $\al n^\be$.

Then the complexity of the adjoint matrix calculation
for the matrix of the order $n=2^p$ by means of binary factorization is
$
{\it C}(n)=\sum_{k=0}^{p-2} 6 \cdot 2^k {\it M}(2^{n-k-1}).
$
We neglect the complexity of multiplication of a matrix by a scalar,
i.e. the terms of the order $n^2$.

Therefore, the asymptotical estimate of the complexity
of the adjoint matrix calculation is $ 3\al n^\be/(1-2^{1-\be})$.

Finally, for the relation of the complexities
of the adjoint matrix calculation and of the matrix multiplication
we obtain the asymptotical estimate
$
k(\be)=\lim_{n \to \infty}{{\it C}(n)\over {\it M}(n)}={3\over 1-2^{1-\be}}.
$
For example we have
$k(3)=4$ for classical multiplication, and $k(\log_2 7)=4.2$ for
Strassen's multiplication.

\section{Linear System Solving in Commutative Domains}
\subsection{ Solving Systems in a Field of Fractions}
Let $\bR$ be a commutative domain with an identity, 
$\bK$ be a field of fractions
of $\bR$, $A \in \bR^{n\times m}$, $\rank A=r$, $c \in \bR^{n}$,
$$
A x = c \eqno(3.1)
$$
be a system of linear equations. 
Let $S$ and $T$ be permutation matrices, which transpose linearly independent
rows and columns of the matrix  $A$ to the upper left corner. We obtain in this corner
a square matrix of size $r\times r$, denote it by $A_0$ ($\det A_0 \neq 0$).
The matrices  $SAT$ and $Sc$ may be written in a block form:
$$
SAT=\mtr {A_0}{A_1}{A_2}{A_3}, \quad 
Sc=\mtrc {c_0} {c_1}, \ \ c_0\in \bR^r.
$$
Denote by
$\cM = \{x \mid x\in\bK^m, Ax=c\}$ the set of all the solutions
of the system (3.1).

We denote by $I_r$ the identity matrix of order $r$,
$E_{i,j}$---square matrices
which have only one nonzero element---$(i,j)$, that equals 1.

We need some facts from the theory of linear equations.

{\bfseries 1.} If $\rank (A,c)\neq r$, then $\cM=\emptyset$.
If $\rank (A,c)= r$, then $\cM$ is a hyperplane of dimension $m-r$ in a space
of dimension $m$. It is defined by $m-r+1$ points, which do not belong
to one hyperplane of lower dimension.

{\bfseries 2.} If the system (3.1) is homogeneous ($c=0$)
and $x_1,x_2,\ldots, x_{m-r}$ are its linearly independent solutions
then $\cM=\{\sum_{i=1}^{m-r+1} x_i u_i \mid u_i\in\bK \}$.

{\bfseries 3.} If the system (3.1) is nonhomogeneous ($c\neq 0$) and
$x_1,x_2,\ldots, x_{m-r+1}$ are its linearly independent solutions
then
$\cM=\{\sum_{i=1}^{m-r+1} x_i u_i \mid
u_i\in\bK,\linebreak
\typeout{linebreak}
\sum_{i=1}^{m-r+1} u_i = 1 \}$.

\begin{defi}
{\em A basis set of solutions of a homogeneous system of linear
equations\/} (3.1) is a set which consists from $m-r$ linearly
independent solutions of the system (3.1).
{\em A basis set of solutions of a nonhomogeneous system of linear
equations\/} (3.1) is a set which consists from $m-r+1$ linearly
independent solutions of the system (3.1).
\end{defi}

The next two theorems reduce the problem of getting the basis solutions
of (3.1) to several problems of solving determined systems.
The first theorem considers homogeneous systems, the second --- nonhomogeneous
systems.

\begin{theorem}
Let (3.1) be a homogeneous system of linear equations and
$A_1=(a_{1},a_{2},\ldots,a_{m-r})$, $a_j\in \bR^r$.
Then the systems
$$
A_0 x_j=-a_j, \ \ \ j=1,\ldots,m-r, \ \eqno(3.2)
$$
are determined and their solutions $x_j\in \bK^r$
define the basis set of solutions of (3.1):
$$
T\mtrc {x_j} {e_j}, \quad  j=1,\ldots,m-r, 
 \eqno(3.3)
$$
where $e_j\in\bR^{m-r}$ are the columns of the identity matrix
$I_{m-r}=\linebreak
\typeout{linebreak}
(e_1,e_2,\ldots,e_{m-r})$.
\end{theorem}

\begin{proof}
Denote $y=T^{-1}x$. By the condition we have
$(A_0,A_1)y=0$. We search for the solution in the form $y=\mtrc {x_j} {e_j}$,
and get the system (3.2).
\end{proof}

\begin{corollary}
Let it be $B=A_0^{-1}A_1=(b_1,b_2,\ldots,b_{m-r})$, $b_j\in \bK^r$.
 Then \linebreak
\typeout{linebreak}
$T\mtrc {-b_1} {e_1}$, $T\mtrc {-b_2} {e_2}$, $\ldots$,
$T\mtrc {-b_{m-r}} {e_{m-r}}$ is the basis set of solutions of (3.1).
\end{corollary}

\begin{theorem}
Let (3.1) be a nonhomogeneous system of linear equations,
$P$ a permutation matrix such that the last element of the vector
$b=PA_0^{-1}c_0$ is not equal to $0$.
Let it be $B=PA_0^{-1}A_1$,
$J\subset\{1,\ldots,m-r\}$ be the numbers of the columns of the matrix 
$B$ with zero elements in the last row.
Let it be  $U=I_{m-r+1}+\sum_{j\in J}E_{1,j+1}$, $W=\di(I_{r-1},U)$, 
$Q=\di(P,I_{m-r})$, $V=QW$,
$I_{m-r+1}=({e'}_0,{e'}_1,\ldots,{e'}_{m-r})$,
where ${e'}_j\in\bR^{m-r+1}$ are the columns of the unit matrix, and 
$({A'}_0,a_0,a_{1},a_{2},\ldots,a_{m-r})$=$P(A_0,A_1)V$
where  ${A'}_0\in \bR^{r\times r-1}$,  ${a}_j\in\bR^{r}$.
Then the systems
$$
({A'}_0,a_j)x_j =Pc_0, \quad  j=0,1,\ldots,m-r,
 \eqno(3.4)
$$
are determined. The solutions of these systems $x_j=\mtrc {{x'}_j}{\xi_j}$,
$ {x'}_j\in \bK^{r-1}$,
$\xi_j\in \bK$ define the basis set of solutions of (3.1):
$$
TV\mtrc {{x'}_j} {\xi_j {e'}_j}, \quad  j=0,1,\ldots,m-r.
 \eqno(3.5)
$$
\end{theorem}

\begin{proof}
Denote $y=V^{-1}T^{-1}x$. By the condition we have
$P(A_0,A_1)Vy=Pc_0$ and $P(A_0, A_1)V=({A'}_0,a_{0},a_{1},\ldots,a_{m-r})$.
If we search for the solution in the form $y=\mtrc {{x'}_j} {\xi_j {e'}_j}$,
then we obtain the systems (3.4).

Let us show that the systems (3.4) are determined and $\xi_j\neq 0$.
Multiply them by $PA_0^{-1}P$ from the left. Since $P=P^{-1}$, we get
$PA_0^{-1}P({A'}_0,a_j)x_j =PA_0^{-1}c_0$, $j=0,1,\ldots,m-r$.

As $b=PA_0^{-1}c_0$ and $PA_0^{-1}P({A'}_0,a_0)=I_r$, the first of the
systems (3.4) gets the form $I_r x_0=b$. Denote $I_r=(I',e)$, $e\in\bR^r$.
We see that  $PA_0^{-1}P{A'}_0=I'$,
$PA_0^{-1}Pa_0=e$.

Since $PA_0^{-1}P(a_0,\ldots,a_{m-r})=PA_0^{-1}P(a_0,PA_1)U=(e,B)U$,
we get $d_j=PA_0^{-1}Pa_j$, $j=1,\ldots,m-r$ are the columns of the matrix
$(e,B)U$. As $U=I_{m-r+1}+\sum_{j\in J}E_{1,j+1}$
and $J$ are the numbers of the columns 
of the matrix $B$ with zero elements in the last row,
the last elements of all columns $d_j$  of the matrix $(e,B)U$ do not equal
zero. The systems (3.4) obtain the form
$$
(I',d_j) x_j=b, \ \ \ j=0,1,\ldots,m-r, \eqno (3.6)
$$
and $\det (I',d_j)\neq 0$. Since the last element of the vector $b$
is nonzero,  solutions  $x_j=\mtrc {{x'}_j}{\xi_j}$ of the systems (3.6) are such
that $\xi_j\neq 0$. So the vectors ${\xi_j {e'}_j}$, $j=1,\ldots,m-r$,
are linearly independent, therefore the vectors (3.5)
 are linearly independent.
\end{proof}

\begin{corollary}
Let it be $B=(b_1,b_2,\ldots,b_{m-r})$,
$b=\mtrc {b'} {\be}$, $b_j=\mtrc {{b'}_j} {\be_j}$;
$\xi_j=\be/\be_j$, $f_j={e'}_j$ for $j\notin J$;
$\xi_j=\be$, $f_j={e'}_j+ {e'}_0$ for $j\in J$.
Then
$$
TQ\mtrc {b'} {\be {e'}_0}, \ \
TQ\mtrc {b'-\xi_j{b'}_j} {\xi_j {f}_j}, \quad 
 j=1,\ldots,m-r, \ \eqno(3.7)
$$
is the basis set of solutions of (3.1).
\end{corollary}

\begin{proof}
Substitute solutions of (3.6) into (3.5). We take into account
that according to the construction,
$d_j=\mtrc {{b'}_j} {\de_j}$, $\de_j=1$
for $j\in J$ and $\de_j=\be_j$ for $j\notin J$. Then we multiply
by the matrix $W$.
\end{proof}

Corollaries 1 and 2 allow to present the algorithm to
construct the basis set of solutions of a system of linear equations for
an arbitrary commutative domain.

\subsection{ System Solving in a Domain}
Now we consider a linear system solving
in a commutative domain. Note that it is not a problem
for homogeneous systems, since any solution in a field of fractions,
been multiplied by a suitable factor, gives the solution in the domain.
So further we shall consider only nonhomogeneous systems.

Let
$\balpha$ be a nonempty finite subset of $\bR$.
The intersection of principal ideals generated by elements of the set
$\balpha$, is the principal ideal $\cup_{p\in \balpha}(p) $.
We denote by  $\LCM(\balpha)$ the generator of
this ideal, i.e. $ (\LCM(\balpha))= \cup_{p\in \balpha}(p) $.

Let $\bK$ be a field of fractions of $\bR$, $x$ be a vector of
the space $\bK^m$, $\alpha_x \subset \bR$ be a set of denominators
of components of $x$.

\begin{defi}
{\it A denominator of a vector} $x$ is
$\chi=\LCM(\alpha_x)$. The vector $x$ will be written as a product
$x=\bx\chi^{-1}$, $\bx\in\bR^m$, $\chi\in\bR$.
\end{defi}

The denominator of a vector $x$ is denoted by ${\DEN}(x)$.

Let $\cM = \{x \mid x\in\bK^m, Ax=b\}$
be the set of all solutions of (3.1)
in $\bK^m$. Denote by $\cM_D=\cM\cup\bR^m$ the set of solutions lying in
the module $\bR^m$. We call $\cM_D$ {\em the set of Diophantine solutions}.

For $x\!=\!(x_1,x_2,\ldots, x_s)$ and
$y\!=\!(y_1,y_2,\ldots, y_s)$ we denote
$\la x,y \ra=\sum_{i=1}^s x_i y_i.$  

\begin{theorem}
Let $\{ x_i=\bx_i\chi_i^{-1} \mid i = 1,2,\ldots, h \}$ be a basis set of
solutions of the nonhomogeneous system (3.1),
$\hbx=(\bx_1, \bx_2, \ldots, \bx_{h} )$,
$\chi=(\chi_1, \chi_2, \ldots, \chi_{h})$.
Then
$$
\cM=\big\{{\la\hbx,q\ra \over \la\chi,q\ra} \bigg|
q\in\bR^{h},\la\chi,q\ra\neq 0\big\}
$$
\end{theorem}

\begin{proof}
Let $u=(u_1,u_2,\ldots, u_{h}) \in \bK^{h}$
be such that $\sum_{i=1}^s u_i=1$. Further we use the notations:
$s_i=u_i\chi_i^{-1}$,
$g$ is $\LCM$ of denominators of all numbers $\{s_i \mid  i=1,2,\ldots, h\}$,
$q_i=g s_i\in \bR$, $i=1,2,\ldots,h$,
$q=(q_1,q_2,\ldots,q_{h})$.

Let it be $\hx=( x_1,  x_2, \ldots,  x_{h} )$,
$z=\la \hx,u \ra$, 
$\sum_{i=1}^s u_i=1$, and so $z \in \cM$.
Then we obtain
$$
z=\la \hx,u\ra =\sum_{i=1}^{h} \bx_i \chi_i^{-1} u_i=
\sum_{i=1}^{h}\bx_i s_i=\la \hbx,q \ra g^{-1}, \eqno(3.7)
$$
$$
1= \sum_{i=1}^s u_i =\sum_{i=1}^{h} \chi_i \chi_i^{-1} u_i=
\sum_{i=1}^{h}\chi_i s_i=\la \chi,q\ra g^{-1}. \eqno(3.8)
$$
So
$\frac{\la\hbx,q\ra}{\la\chi,q\ra}=z\in \cM$.

Conversely, if $q\in \bR^{h}$,
$\la\chi,q\ra=g\neq 0$ and $z={\la\hbx,q\ra \over \la\chi,q\ra}$,
then, denoting $u_i=\chi_i q_i g^{-1}$, obtain
$z=\la \hx,u\ra$ and $\sum_{i=1}^s u_i=1$.
Therefore $z\in \cM$.
\end{proof}

\begin{corollary} The set $I_A$=$\{{\DEN}(x) \mid Ax=c, x\in \bK^m\}\cup 0$
is an ideal in $\bR$. 
\end{corollary}

\begin{corollary} $\cM_D\neq \emptyset$ if and only if $I_A=\bR$.
\end{corollary}

\begin{corollary} A system (3.1) has Diophantine solutions if 
the ideal generated by the denominators of a basis set of solutions is unit.
\end{corollary}

\begin{corollary}
Let the ideal, generated by the denominators of basis solutions
$x_i=\bx_i \chi_i^{-1}$, $i=1,2,\ldots,h$ of  the system (3.1) be unit,
$\chi=(\chi_1,\chi_2,\ldots,\chi_{h})$. Then there exists a nonzero vector 
$q=(q_1,q_2,\ldots,q_{h}) \in \bR^{h}$, such that 
$\la \chi,q \ra=1$ and $\la\bx,q \ra$ is Diophantine solution of (3.1).
If in addition  $q_s\neq 0$, then 
$$
 x_1,\ldots,x_{s-1}, \la\bx,q \ra, x_{s+1},\ldots,x_{h}
$$
is a basis set of solutions for this system.
\end{corollary}

\begin{defi}
We call a~{\em Diophantine basis of solutions for a system\/}
$A x = b$ a basis
set of solutions for this system, that wholly belongs to $\bR^m$.
\end{defi}

In other words a Diophantine basis consists of the $m+1-r$ linearly independent 
solutions of nonhomogeneous ($m-r$ for homogeneous) system, that belong to
$\bR^m$.

\begin{corollary}
Let $x_i=\bx_i \chi_i^{-1}$, $i=1,2,\ldots,h$ be a basis set of solutions
of (1) and $\chi_1=1$.
Then the set $\bx_1$, $\bx_i-\bx_1(\chi_i-1)$, $i=2,3,\ldots,h$, 
is a Diophantine basis of solutions for this system.
\end{corollary}

\begin{proof}:
Since all $z_i=\bx_i-\bx_1(\chi_i-1)$, $i=2,3,\ldots,h$ belong to
$\bR^{m}$ and are linearly independent together with $x_1$, it remains
to show that $z_i$ are the solutions of (1). As $\chi_1=1$, we have
$$
z_i=\bx_i-\bx_1(\chi_i-1)=
{\bx_1(-\chi_i+1)+\bx_i\cdot 1 \over \chi_1 (-\chi_i+1)+\chi_i\cdot 1 }.
$$
Therefore, by Theorem 3, $z_i$, $i=2,3,\ldots,h$ are the solutions of
the system $Ax=b$.
\end{proof}

Corollaries of Theorem 5 allow to present a rendomized algorithm for
computing a Diophantine basis of a system solutions.

If the ideal generated by the denominators of a basis set of solutions is unit,
than we compute a Diophantine basis. Else we must choose a new permutation
matrices $S$ and $T$ and compute a new rational basis, and so on.
So we get some kind of iterative Diophantine solve. 
As it is proved in \cite{MulStor99}
an expected number of rational solutions that are necessary for 
getting Diophantine solution is $O(\log n + \log \log \|(A,c)\|)$
for the ring $\Bbb Z$. One of the question here is how to get such evaluation
for another rings.

As we get $m-n+1$ rational solutions in one iteration for $\rank A=n$,
so an expected number of iterations for the ring $\Bbb Z$  is  
$O((m-n+1)^{-1}(\log n + \log \log \|(A,c)\|))$.

\section{ Solving Linear System in Principal Ideal Domain}
\subsection{Solving Linear System in the Field of Fractions}
Theorems 3 and 4 reduce the problem of getting the basis solutions
of the system (3.1) to several problems of solving determined systems.
Computing a basis set of solutions of nonhomogeneous system needs
to solve $m-r+1$ systems with the matrix of coefficients of size
$r\times r$ (the case of homogeneous systems needs to solve $m-r$ systems).
To solve a determined system one may use the p-adic lifting.

Recall a general scheme of p-adic methods. The suitable prime element
$p$ of the ring $\bR $ is chosen. The element $p$ must not divide
the determinant of the coefficients matrix. For example,
in the ring $\Bbb Z$ this choice may appear to be unsuccessful
with probability no more than $1/p$.
If the check shows that the solution is incorrect, then another
prime element  $p$ is chosen. The ring of residues prime modulo
$p$ is a field, and the solution in this field may be found
for example using the Gauss method. The upper evaluations for
numerators and denominators of system solutions is calculated
by means of Hadamard inequality. According to them the upper evaluation
for $p^k$ --- the boundary of lifting. Then the solution is lifted
modulo $p$ up to $p^k$ and a rational solution is constructed.
For solving one determined system we apply the algorithm given
by Dixon \cite{Dixon1982}, which use a linear p-adic lifting.
Its complexity is $O(n^3 (\log n + \log \|A\| +\log p)^2 )$ bit operations.
So the algorithm for getting all the rational solutions of the 
system $Ax=c$, when $\rank A=n$, has the complexity
$O(n^2m(m-n+1)(\log n + \log \|A\| +\log p)^2 )$ bit operations.

\subsection{Solving Linear System in the Principal Ideal Domain}
A one iterative step for computing a Diophantine basis is consist in
the computing a rational basis and a Diophantine solutions,
due to the Corollaries of Theorem 5.

The complexity of the first is
$O(mn^{\beta-1})$, of the second --- $O(m(m-r)+C_G)$
operations in the ring $\bR$. $C_G$ is the amount of operations
that are necessary for an expansion of a unit in $m-r$ generators of
the unit ideal in the ring $\bR$.
Such expansion of a unit may be obtained for example with the help of
the algorithm of computing the Gr\"obner basis.
The evaluation of the complexity of such algorithm in general is
not the subject of this paper. Mention that for Euclidean
rings $C_G$ it is the amount of operations in  Euclidean
algorithm, that calculates the GCD for $m-r$ numbers.

The more defined evaluations may be obtained for algorithms,
using the linear p-adic lifting. It is known that the complexity
of Dixon algorithm \cite{MulStor99} in the case of integer numbers $\Bbb Z$,
for $n=m$ and $\det A\neq 0$ is bounded by the number
$\cD_Z=O(n^3(\log n +\log \|A\| +\log p)^2 +n(\log \|á\|)^2)$
bit operations.
In the case of the ring of polynomials $F[x]$ over a field $F$
this complexity is bounded by
$ \cD_{F[x]}= O(n^3(\| A \| +\| p \|)^2 +n(\| c \|)^2)$
operations in the field $F$.
The function $\| \ \ \|$ has the next meaning:
$\|\alpha\|=|\alpha|$ for $\alpha\in {\Bbb Z}$,
$\|\alpha\|={\deg \,} \alpha$ for $\alpha\in F[x]$, for matrix $A$,
$\|A\|=\max_{\alpha\in A}\|\alpha\|$.

Computing a basis from $m-r$ solutions needs not more than
$\cC_Z=(m-r)\cD_Z$ and $\cC_{F[x]}=(m-r)\cD_{F[x]}$
operations for each case correspondingly.

The average number of matrix inversions for computing 
one rational solution is now $m-r$ times less 
than this number in the algorithm \cite{MulStor99}.

The complexity of computing a Diophantine basis
for $\Bbb Z$ and $F[x]$
is the same as the complexity of computing a
rational basis. An expected number of iterations is  
$O((m-n+1)^{-1}(\log n + \log \log \|(A,c)\|))$.

\end{document}